\documentclass{article}
\usepackage{arxiv}
\usepackage[utf8]{inputenc}
\usepackage{listings}
\usepackage{xcolor}
\usepackage{xspace}
\usepackage{tikz,calc}
\usetikzlibrary{automata,arrows,positioning,calc,shapes.multipart,snakes}
\tikzset{initial text={},every initial by arrow}
\usetikzlibrary{chains,shadows.blur}
\usepackage{thmtools,thm-restate}
\usepackage[T1]{fontenc}    
\usepackage{hyperref}       
\usepackage{url}   
\usepackage{amsmath,amssymb,amsthm}   
\usepackage{pgfplots}
\usepackage{multirow}
\usepackage{colortbl}
\usepackage{todonotes}

\newcommand{\ie}{\emph{i.e.,}\xspace}
\newcommand{\eg}{\emph{e.g.,}\xspace}
\newcommand{\etal}{\emph{et al.}\xspace}
\newtheorem{definition}{Definition}
\newtheorem{example}{Example}

\newcommand{\names}{\mathcal{N}}
\newcommand{\symbols}{\mathtt{Sym}}
\newcommand{\val}{\mathtt{Val}}
\newcommand{\var}{\mathtt{Var}}

\definecolor{mGray}{rgb}{0.5,0.5,0.5}
\definecolor{backgroundColour}{rgb}{0.95,0.95,0.92}
\lstdefinestyle{CStyle}{
    backgroundcolor=\color{backgroundColour},   
    commentstyle=\itshape\color{purple!40!black},
    keywordstyle=\bfseries\color{green!40!black},
    numberstyle=\tiny\color{mGray},
    stringstyle=\color{orange},
    basicstyle=\footnotesize,
    breakatwhitespace=false,         
    breaklines=true,                 
    captionpos=b,                    
    keepspaces=true,                 
    numbers=left,                    
    numbersep=3pt,                  
    showspaces=false,                
    showstringspaces=false,
    showtabs=false,                  
    tabsize=2,
    belowcaptionskip=1\baselineskip,
    xleftmargin=\parindent,
    escapeinside={*@}{@*},
    language=C
}
\hypersetup{
    colorlinks=true, 
    linktoc=all,      
    linkcolor=blue,  
}

\newcommand{\ttt}[1]{\texttt{#1}}
\newcommand{\mtt}[1]{\mathtt{#1}}

\newcommand{\B}{\mtt{B}}

\newcommand{\oprel}{\:\mtt{oprel}\:}
\newcommand{\opreli}{\:\mtt{oprel^i}\:}

\definecolor{pblue}{rgb}{0.13,0.13,1}
\definecolor{pgreen}{rgb}{0,0.5,0}
\definecolor{pred}{rgb}{0.9,0,0}
\definecolor{pgrey}{rgb}{0.46,0.45,0.48}

\lstdefinestyle{JStyle}{
backgroundcolor=\color{backgroundColour}, 
language=Java,
  showspaces=false,
  showtabs=false,
  breaklines=true,
     numbers=left,                    
    numbersep=3pt,
  showstringspaces=false,
  commentstyle=\color{pgreen},
  keywordstyle=\color{pblue},
  stringstyle=\color{pred},
      basicstyle=\footnotesize,
    breakatwhitespace=false,  
  moredelim=[il][\textcolor{pgrey}]{$$},
  moredelim=[is][\textcolor{pgrey}]{\%\%}{\%\%}
}

\newcommand{\squishlist}{
 \begin{list}{-}
  { \setlength{\itemsep}{0pt}
     \setlength{\parsep}{1pt}
     \setlength{\topsep}{1pt}
     \setlength{\partopsep}{0pt}
     \setlength{\leftmargin}{0.6em}
     \setlength{\labelwidth}{1.5em}
     \setlength{\labelsep}{0.4em} } }
\newcommand{\squishend}{
  \end{list}  }
\title{Detecting Fault Injection Attacks with Runtime Verification}

\author{
  Ali~Kassem\\
  Univ. Grenoble Alpes, Inria, 
  \\CNRS, Grenoble INP, LIG\\
  38000 Grenoble\\
  France \\
  \texttt{ali.kassem@inria.fr} \\
  \And
  Yli\`{e}s~Falcone\\
  Univ. Grenoble Alpes, Inria, 
  \\CNRS, Grenoble INP, LIG\\
  38000 Grenoble\\
  France \\
  \texttt{ylies.falcone@inria.fr} \\
}
\begin{document}
\maketitle
\begin{abstract}
Fault injections are increasingly used to attack/test secure applications. 
In this paper, we define formal models of runtime monitors that can detect fault injections that result in test inversion attacks and arbitrary jumps in the control flow. 
Runtime verification monitors offer several advantages.
The code implementing a monitor is small compared to the entire application code.  
Monitors have a formal semantics; and we prove that they effectively detect attacks.
Each monitor is a module dedicated to detecting an attack and can be deployed as needed to secure the application.
A monitor can run separately from the application or it can be ``weaved'' inside the application. 
Our monitors have been validated by detecting simulated attacks on a program that verifies a user PIN.  
\end{abstract}


\section{Introduction}
%
%
Fault injections are effective techniques to exploit vulnerabilities in embedded applications and implementations of cryptographic primitives~\cite{eurocrypt/BonehDL97,fdtc/BalaschGV11,cosade/DehbaouiMMDT13,cardis/KumarBBGV18,IEEEares/BerthomeHKL12,jhss/YuceSW18}. 
Fault injections can be thwarted (or detected) using software or hardware countermeasures~\cite{tvlsi/KaraklajicSV13,pieee/Bar-ElCNTW06}.
Hardware countermeasures are expensive and unpractical for off-the-shelf products. 
Henceforth, software countermeasures are commonly adopted. 
Software countermeasures can be categorized into algorithm-level and instruction-level countermeasures. 

Most of the existing algorithm-level countermeasures make use of redundancy such as computing a cryptographic operation twice then comparing the outputs~\cite{joc/BonehDL01}, or using parity bits~\cite{ches/KarriKG03}. 
Algorithm-level countermeasures are easy to implement and deploy and usually introduce low (computational and footprint) overhead. 
However, as they are usually based on redundancy, they can be broken using multiple faults injection or by skipping critical instruction~\cite{ches/AumullerBFHS02,cosade/EndoHHTFA14}. 

In the other hand, instruction-level countermeasures require changes to the low level instruction code, for example by applying instruction duplication or triplication~\cite{cases/BarenghiBKPR10,hipeac/BarryCR16}. 
Instruction-level countermeasures are more robust than algorithm-level countermeasures.
Indeed, it is believed that instruction-level countermeasures are secure against multiple faults injection under the assumption that skipping two consecutive instructions is too expensive and requires high synchronization capabilities~\cite{cases/BarenghiBKPR10,jce/MoroHER14}. 
However, instruction-level countermeasures introduce a large overhead, for instance, instruction duplication doubles the execution time.
Moreover, instruction-level countermeasures require changing the instruction set (with dedicated compilers) since \eg some instructions have to be replaced by a sequence of idempotent or semantically equivalent instructions.

In this paper, we use runtime verification principles~\cite{HavelundG05,LeuckerS09,FalconeHR13,BartocciFFR18} and monitors to detect fault injections that result in test inversion or unexpected jumps in the control flow of a program. 
A monitor can run off-line after the end of an execution provided that the necessary events have been saved in a safe log, or on-line in parallel with the execution. 
We use Quantified Event Automata (QEAs)~\cite{fm/BarringerFHRR12} to express our monitors.
We prove that our monitors for test inversion and jump attacks detect them if and only if such attack occur at runtime -- Propositions~\ref{prop:test_inversion} and~\ref{prop:jump}, respectively. 
From an implementation point of view, we validate Java implementations of our monitors using attack examples on a program that verifies a user PIN code taken from the FISSC benchmark~\cite{Dureuil2016}.
Our monitors are lightweight and small in size. The monitors execution time is proportional to the size of the program under surveillance. 
However, the memory overhead can be bounded as only data related to the ``active'' basic block need to be kept in memory. 

This make them suitable for low-resource devices where monitors can run in a small (hardware)-secured memory (where the entire application may not fit). 
Moreover, monitors can run separately from the application under surveillance, and thus can run remotely (in a secure environment) provided a secure communication channel.

The rest of the paper is structured as follows.
Section~\ref{sec:QEAs} overviews Quantified Event Automata. 
Section~\ref{sec:preliminaries} introduces preliminaries. 
Section~\ref{sec:modeling} introduces execution and attacker models.
Section~\ref{sec:monitors} introduces our monitors.
Section~\ref{sec:casestudy} describes an experiment that validates the effectiveness of our monitors.  
Section~\ref{sec:related-work} discusses related work. 
Finally, Section~\ref{sec:conclusion} concludes and outlines avenues for future work.  
\section{Quantified Event Automata}
\label{sec:QEAs}
We briefly overview Quantified Event Automata~\cite{fm/BarringerFHRR12} (QEAs) which are used to express monitors.  
QEAs are an expressive formalism to represent parametric specifications to be checked at runtime. 
An Event Automaton (EA) is a (possibly non-deterministic) finite-state automaton whose alphabet consists of parametric events and whose transitions may be labeled with guards and assignments.
The syntax of EA is built from a set of event names $\names$, a set of values $\val$, and a set of variables $\var$ (disjoint from $\val$). 
The set of symbols is defined as $\symbols = \val \cup \var$. 
An event is a tuple $\langle \mtt{e, p_1, \ldots, p_n} \rangle$, where $\mtt{e} \in \names$ is the event name and $\mtt{p_1, \ldots, p_n} \in \symbols^{\mtt{n}}$ are the event parameters. 
We use a functional notation to denote events: $\langle \mtt{e, p_1, \ldots, p_n} \rangle$ is denoted by $\mtt{e(p_1, \ldots, p_n)}$. 
Events that are variable-free are called ground events, \ie an event $\mtt{e(p_1, \ldots, p_n)}$ is ground if $\mtt{p_1, \ldots, p_n} \in \val^{\mtt{n}}$. 
A \emph{trace} is defined as a finite sequence of ground events.  
We denote the empty trace by $\epsilon$. 

The semantics of an EA is close to the one of a finite-state automaton with the natural addition of guards and assignments on transitions.
A transition can be triggered only if its guard evaluates to $\mathtt{True}$ with the current binding (a map from variables to concrete values), and the assignment updates the current binding. 

A QEA is an EA with some (or none) of its variables quantified by $\forall$ or $\exists$. 
Unquantified variables are left free and they can be manipulated through assignments and updated during the processing of the trace. 
A QEA accepts a trace if after instantiating the quantified variables with the values derived from the trace, the resulting EAs accept the trace. 
Each EA consumes only a certain set of events, however a trace can contain other events which are filtered out. 
The quantification $\forall$ means that a trace has to be accepted by all EAs, while the quantification $\exists$ means that it has to be accepted by at least one EA.  
For a QEA $\mtt{M}$ with quantified variables $\mtt{x_1,\ldots,x_n}$. We use the functional notation $\mtt{M(x_1, \ldots, x_n)}$ to refer to the related EAs depending on the values taken by $\mtt{x_1,\ldots,x_n}$. 

We depict QEAs graphically.
The initial state of a QEA has an arrow pointing to it. 
The shaded states are final states (\ie accepting states), while white states are failure states (\ie non-accepting states). 
Square states are closed-to-failure, {\ie} if no transition can be taken then there is an implicit transition to a (sink) failure state. 
Circular states are closed-to-self (aka skip-states), {\ie} if no transition can be taken, then there is an implicit self-looping transition.
We use the notation $\frac{\mathit{guard}}{\mathit{assignment}}$ to write guards and assignments on transitions, $:=$ for variable assignment, and $==$ for equality test.
\begin{figure}[t]
\centering
\begin{tikzpicture}
\tikzstyle{main}=[circle,thick,draw,fill=white,minimum size=6mm]
\tikzstyle{failure}=[rectangle,thick,draw,fill=white,minimum size=6mm]
\node[] (0) at (-1,0.8)  {$\forall \mtt{i}$};

\node[initial, failure, fill= gray] (1) at (0,0.2)  {1};
\node[main, fill= gray] (2) at (3,0.2)  {2};
\draw[->] (1) to node [above] {$\mtt{e_1(i)}$} (2);
\draw[thick] ($ (current bounding box.south west) + (0,-0.3)$) rectangle  ($ (current bounding box.north east) + (0.3,0)$); 
\end{tikzpicture}
\caption{$\mtt{M_R}$, a QEA for requirement $\mtt{R}$ from Example~\ref{ex:qea}.}
\label{fig:example}
\end{figure}
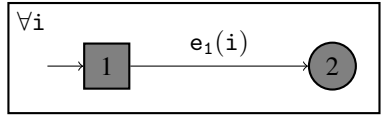

\begin{example}[QEA]
\label{ex:qea}
Figure~\ref{fig:example} depicts\ $\mtt{M_R}$, a QEA that checks whether a given trace satisfies requirement $\mtt{R}$: ``for all $\mtt{i}$, event $\mtt{e_1(i)}$ should precede event $\mtt{e_2(i)}$''.
The alphabet of $\mtt{M_R}$ is $\Sigma_{\mtt{R}} = \{\mtt{e_1(i)},\mtt{e_2(i)}\}$. 
Consequently, any event in an input trace that is not an instantiation of $\mtt{e_1(i)}$ or $\mtt{e_2(i)}$ is ignored by $\mtt{M_R}$. 
QEA $\mtt{M_R}$ has two states (1) and (2), which are both accepting states, one quantified variable $\mtt{i}$, and zero free variables. 
As its initial state (1) is an accepting state, the empty trace is accepted by $\mtt{M_R}$.
State (1) is a square state, hence an event $\mtt{e_2(i)}$ at state (1) leads into an implicit failure state, which is what we want as it would not be preceded by event $\mtt{e_1(i)}$ in this case. 
An event $\mtt{e_1(i)}$ at state (1) leads into state (2) which is a skipping state, so after an occurrence of event $\mtt{e_1(i)}$ any sequence of events (for the same value of $\mtt{i}$) is accepted.
We note that one can equivalently replace state (2) with an accepting square state with a self-loop labeled by $\Sigma_{\mtt{R}}$. 

Quantification $\forall \mtt{i}$ means that the property must hold for all values that $\mtt{i}$ takes in a trace.
Each instantiation of $\mtt{i}$ results in an EA.  

To decide whether a trace is accepted by $\mtt{M_R}$ or not, the trace is first sliced based on the values that can match $\mtt{i}$.
Then, each slice is checked against the event automaton instantiated with the appropriate value for $\mtt{i}$. 
For instance, the trace $\mtt{e_1(I_1)}.\mtt{e_2(I_2)}.\mtt{e_2(I_1)}.\mtt{e_1(I_2)}$ is sliced into the following two slices:
%
	 $\mtt{i}\mapsto \mtt{I_1}: \mtt{e_1(I_1)}.\mtt{e_2(I_1)}$, and  
	 $\mtt{i}\mapsto \mtt{I_2}: \mtt{e_2(I_2)}.\mtt{e_1(I_2)}$.  
%
Each slice is checked independently. 
The slice associated with $\mtt{I_1}$ is accepted by $\mtt{M_{R}(I_1)}$ as it ends in the accepting state (2), while the slice associated with $\mtt{I_2}$ is rejected by $\mtt{M_{R}(I_2)}$ since the event $\mtt{e_2(I_2)}$ at state (1) leads to an implicit failure state.
Therefore, the whole trace is rejected by the QEA because of the universal quantification on $\mtt{i}$. 
\end{example}

\section{Preliminaries and Notations}
\label{sec:preliminaries}
As a running example, we consider function \ttt{verifyPIN} which is depicted in Listing~\ref{lst:verifyPIN}\footnote
{
The code is inspired from the C version in the FISSC benchmark~\cite{Dureuil2016}.
}.  
Function \ttt{verifyPIN} is the main function for the verification of a user PIN.
It handles the counter of user trials (variable \ttt{g\_ptc}), which is initialized to 3, \ie the user is allowed for 3 trials (one trial per execution).  
The user is authenticated if the value of \ttt{g\_ptc} is greater than 0, and function \ttt{byteArrayCompare} returns \ttt{BOOL\_TRUE}.
Function \ttt{byteArrayCompare} returns \ttt{BOOL\_TRUE} if \ttt{g\_userPin} and \ttt{g\_cardPin} are equal. 
Note that \ttt{BOOL\_TRUE} and \ttt{BOOL\_FALSE} have the values \ttt{0xAA} and \ttt{0x55}, respectively.  
This provides a better protection against faults that modifies data-bytes.

\begin{lstlisting}[style=CStyle,
	caption={Code of \ttt{verifyPIN}.},
	captionpos=t,
	label=lst:verifyPIN,
	float=tp
	]
void verifyPIN() {
 g_authenticated = BOOL_FALSE;
 if(g_ptc > 0) {
  if(byteArrayCompare(g_userPin,g_cardPin) == BOOL_TRUE) {
   g_ptc = 3;     /*reset the counter of trials*/
   g_authenticated = BOOL_TRUE; }
  else 
   { g_ptc--; }   /*one trial less remaining*/
 } 
 return g_authenticated; 
}
\end{lstlisting}

The monitors defined in Section~\ref{sec:monitors} are generic.
They are independent from programming language and do not require changes to the low-level code because the required instrumentation to produce events can be made at the the source code level.   
However, to describe attacks at a lower level, we make use of the three-address code (TAC) representation. 
TAC is an intermediate-code representation which reassembles for instance LLVM-IR.
TAC is machine independent, easy to generate from source code, and can be easily converted into assembly code.
In TAC, a program is a finite sequence of three-address instructions. 
In particular, an instruction $\mtt{\mbox{\ttt{ifZ} } z \mbox{ \ttt{goto} } L}$ is a conditional branch instruction that directs the execution flow to $\mtt{L}$ if the value of $\mtt{z}$ is 0 (\ie false). 
An instruction \ttt{goto L} is an unconditional branch instruction that directs the execution flow to $\mtt{L}$.  
A label $\mtt{L}$ can be assigned to any instruction in the TAC. 
Instruction $\mbox{\ttt{Push} } \mtt{x}$ pushes the value of $\mtt{x}$ onto the stack.
Before making a function call, parameters are individually pushed onto the stack from right to left.
While $\mbox{\ttt{Pop} } \mtt{k}$ pops $\mtt{k}$ bytes from the stack; it is used to pop parameters after a function call. 

\begin{lstlisting}[style=CStyle,caption={The TAC representation of \ttt{verifyPIN}.},captionpos=t,label=lst:verifyPIN-tac,
	float=bp]
verifyPIN: 
 g_authenticated := BOOL_FALSE 
 _t0 := (g_ptc > 0) 
 ifZ _t0 goto L1
 Push g_cardPin
 Push g_userPin
 _t1:= call byteArrayCompare
 Pop 64
 _t2 := (_t1 == BOOL_TRUE)
 ifZ _t2 goto L2
 g_ptc := 3                  
 g_authenticated := BOOL_TRUE 
 goto L1
 L2: g_ptc := -g_ptc                 
 L1: return g_authenticated
\end{lstlisting}

\begin{example}
Listing~\ref{lst:verifyPIN-tac} depicts the TAC representation of \ttt{verifyPIN}. 
The variables \ttt{\_t0}, \ttt{\_t1} and \ttt{\_t2} are compiler-generated temporaries. 

\end{example}

To specify how the program under verification has to be instrumented, in Section~\ref{sec:execution}, we refer to the control flow graph (CFG) of the program. 
The CFG of a program is a representation of all the paths that might be taken during its execution. 
A CFG is a rooted directed graph ($\mtt{V, E}$), where $\mtt{V}$ is a set of nodes representing basic blocks, and $\mtt{E}$ is a set of directed edges representing possible control flow paths between basic blocks.
A CFG has an entry node and one exit node.
The entry node has no incoming edges while the exit node has no outgoing edges.
A basic block is a maximal sequence $\mtt{S_1 \ldots S_n}$ of instructions such that 
\begin{itemize} 
  \item  
  it can be entered only at the beginning, \ie none of the instructions $\mtt{S_2 \ldots S_n}$ has a label (\ie target of a branch), and 
  \item 
  it can be exited only at the end, \ie none of the 
  instructions $\mtt{S_1 \ldots S_{n-1}}$ is a branch instruction or a return. 
\end{itemize}
A CFG may contain loops. A loop is a sequence $\mtt{B_1, \ldots, B_n}$ of basic blocks dominated by the first basic block: $\mtt{B_1}$, and having exactly one back-edge from the last basic block: $\mtt{B_n}$ into $\mtt{B_1}$.
Note that in a CFG, a basic block $\mtt{B}$ dominates a basic block $\mtt{B'}$ if every path from the entry node to $\mtt{B'}$ goes through $\mtt{B}$. An edge $\mtt{(B',B)}$ is called a back-edge if $\mtt{B}$ dominates $\mtt{B'}$. 

\begin{example}
Figure~\ref{fig:verifyPIN-cfg} depicts the CFG of \ttt{verifyPIN}. 
\end{example}

\begin{figure}[t]
\centering
{\small 
\begin{tikzpicture}
[
  auto,
  node distance = 10mm,
  basic/.style = {draw,rounded corners, align=left, minimum height =.75cm, minimum width=1.5cm, node distance=.65cm},
  box/.style = {basic,rectangle split,rectangle split parts=2, rectangle split part fill={lightgray,white}}, 
]
 \node[box] (b1) 
   {
    $\mtt{B_1}$ 
    \nodepart{second} 
    \ttt{g\_authenticated := BOOL\_FALSE} \\
    \ttt{\_t0 := (g\_ptc > 0)} \\ 
    \ttt{ifZ \_t0 goto L1}
   }; 
 \node[box,below=of b1] (b2) 
   {
    $\mtt{B_2}$ 
    \nodepart{second} 
    \ttt{Push g\_cardPin} \\ 
    \ttt{Push g\_userPin} \\ 
    \ttt{\_t1:= call byteArrayCompare} \\ 
    \ttt{Pop 64} \\ 
    \ttt{\_t2 := (\_t1 == BOOL\_TRUE)} \\ 
    \ttt{ifZ \_t2 goto L2}
   };
 \node[box,below=of b2] (b3) 
   {
    $\mtt{B_3}$ 
    \nodepart{second} 
    \ttt{g\_ptc := 3} \\ 
    \ttt{g\_authenticated := BOOL\_TRUE} \\ 
    \ttt{goto L1}
   }; 
 \node[box,right=of b3] (b4) 
   {
    $\mtt{B_4}$ 
    \nodepart{second} 
    \ttt{L2:g\_ptc := -g\_ptc}
   };  
 \node[box,below=of b3] (bf) (b5) 
   {
    $\mtt{B_5}$ 
    \nodepart{second} 
    \ttt{L1:return g\_authenticated}
   };  
 \draw[->,thick] (b1) -- node{T} (b2);
 \draw[->,thick] (b1.west) -- ++(-0.5,0) -| ++(0,-4) node[xshift=0.2cm]{F} |- (b5.west);  
 \draw[->,thick] (b2) -- node{T} (b3);   
 \draw[->,thick] (b2.east) -- ++(1.5,0) node[yshift=0.2cm]{F} -| (b4);   
 \draw[->,thick] (b3) -- (b5);   
 \draw[->,thick] (b4) -- ++(0,-1.5) |- (b5.east);   
\end{tikzpicture}
}
\caption{CFG for the \ttt{verifyPIN} function.}
\label{fig:verifyPIN-cfg}
\end{figure}
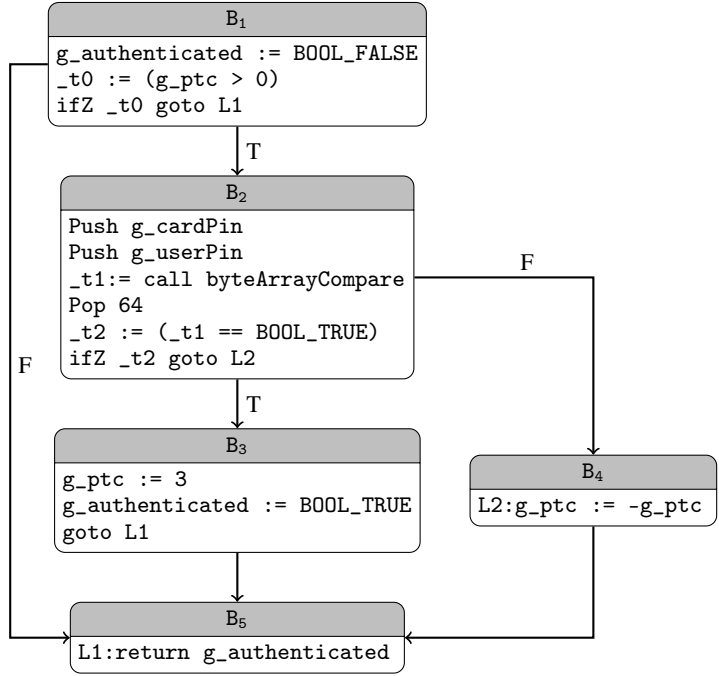

We define a \emph{test}, in TAC, as an instruction that involves a logical expression, \ie an instruction of the form $\mtt{z:=(x \oprel y)}$ where $\oprel$ is a logical operator, followed by the conditional branch instruction $\mtt{\mbox{\ttt{ifZ} } z \mbox{ \ttt{goto} } L}$ for some label \ttt{L}. For a test $\mtt{A}$ we use the notation $\mtt{cond(A)}$ to refer to the condition (\ie the logical expression) involved in $\mtt{A}$. 
\begin{example}[Test]
Function \ttt{verifyPIN} contains two tests:
\begin{itemize}
\item	
\ttt{\_t0 := (g\_ptc > 0)}, \ttt{ifZ \_t0 goto L1}  
\item	
\ttt{\_t2 := (\_t1 == BOOL\_TRUE)}, \ttt{ifZ \_t2 goto L2}. 
\end{itemize}
\end{example}
Let $\B$ be a basic block. We use the notation $\mtt{B_L}$ to refer to the successor of $\B$ whose head is labeled by \ttt{L} if the tail of $\B$ is an unconditional jump instruction \ttt{goto L}. 
Similarly, we use the notation $\mtt{B_T}$ (\emph{resp.} $\mtt{B_F}$) to refer to 
the successor of $\B$ that is executed when $\mtt{(x \oprel y)=True}$ (\emph{resp.} $\mtt{(x \oprel y)=False}$) if $\B$ ends with a test $\mtt{A}$ where $\mtt{cond(A)}=\mtt{(x \oprel y)}$. 
Note that $\mtt{B_F}$ is the basic block whose head is labeled by \ttt{L}. 

Finally, in what follows, variable $\mtt{i}$ is used to denote the unique identifier of a basic block $\B$, and $\mtt{i_{\square}}$ is used to denote the unique identifier of a basic block $\B_{\square}$. 

\section{Modeling}
\label{sec:modeling}
We define our execution model in Section~\ref{sec:execution}. 
Then, in Section~\ref{sec:attacker-model}, we define an attacker model for test inversion and jump attacks. 
\subsection{Modeling Execution} 
\label{sec:execution} 
We define a program execution (a program run) by a finite sequence of events, a 
trace. Such event-based modeling of program executions is appropriate for monitoring actual events of the program.  
Events mark important steps in an execution. We consider parametric events
of the form $\mtt{e(p_1, \ldots, p_n)}$, where $\mtt{e}$ is the event name, and $
\mtt{p_1, \ldots, p_n}$ is the list of symbolic parameters that take some concrete  values at runtime. 

As we consider fault injection attacks, then events themselves are under threat  
(\eg an attacker may skip an event emission). 
Skipping an event may result in a monitor reporting a false attack. 
In order to ensure events emission, we assume that the program under verification is instrumented so that every event is consecutively emitted twice. 

We define the following events which have to be emitted, two consecutive times each, during a program execution, where $\mtt{G}$ = ($\mtt{V, E}$) is the CFG of the program: 
\begin{itemize}
  \item
  For every basic block $\mtt{\B\in V}$, event $\mtt{begin(i)}$ has to be emitted at the beginning of $\B$. 
  \item
  For every basic block $\mtt{\B\in V}$, event $\mtt{end(i)}$ has to be emitted:  
      \begin{itemize}
        \item 
        just before instruction \ttt{return}, if the tail of $\B$ is \ttt{return}. 
        \item 
        at the beginning of $\mtt{B_L}$, if the tail of $\B$ is an unconditional jump instruction \ttt{goto L}. 
        Note that, in this case, events $\mtt{end(i)}$ have to be emitted before event $\mtt{begin(i_L)}$. 
        \item 
        at the beginning of both $\mtt{B_T}$ and $\mtt{B_F}$, if the tail of $\B$ is a conditional jump instruction $\mtt{\mbox{\ttt{ifZ} } z \mbox{ \ttt{goto} } L}$.   
        Note again that, in this case, events $\mtt{end(i)}$ have to be emitted before events $\mtt{begin(i_T)}$ and $\mtt{begin(i_F)}$. 
        \item 
        at the end of $\B$, otherwise.  
      \end{itemize}
  \item 
  For every loop 
  $\mtt{B_1, \ldots, B_n}$ in graph $\mtt{G}$, 
  events $\mtt{reset(i_1), \ldots}$, $\mtt{reset(i_n)}$ have to be emitted at the end of $\mtt{B_n}$. 
  Event $\mtt{reset(i)}$ means that the basic block whose identifier is $\mtt{i}$ may be executed again as it is involved in a loop.  
  Note that, in this case, $\mtt{reset}$ events have to be emitted before event $\mtt{end(i_n)}$ as the latter is used to detect jump attacks.  
  \item 
  For every basic block $\mtt{\B\in V}$ that ends with a test $\mtt{A}$,  
  events $\mtt{bT(i,x,y)}$ and $\mtt{bF(i,x,y)}$ have to be emitted at the beginning of $\mtt{B_T}$ and $\mtt{B_F}$, respectively, where $\mtt{cond(A)}=\mtt{(x \opreli y)}$. 
  We note that, in this case, the identifier $\mtt{i}$ also identifies the test $\mtt{A}$ and the logical operator $\opreli$ as a basic block can contain at most one test.  
\end{itemize}

We define a program execution as follows. 
\begin{definition}{\bf (Program Execution).}
Let $\mtt{P}$ be a program. 
An execution $\mtt{P_{exec}}$ of $\mtt{P}$ is a finite sequence of events $\mtt{e_1. \cdots . e_n}$, where $\mtt{n}\in\mathbb{N}$, such that $\mtt{e_j}\in\Sigma_{\mtt{ALL}}$ = $\{\mtt{begin(i)}$, $\mtt{end(i)}$, $\mtt{reset(i)}$, $\mtt{bT(i,x,y)}$, $\mtt{bF(i,x,y)}\}$ for every  $\mtt{j \in \{1, \ldots, n\}}$.
\end{definition}
For an execution $\mtt{P_{exec}}$, we use the functional notation $\mtt{P_{exec}(i)}$ to refer to the trace obtained from $\mtt{P_{exec}}$ by considering only the related events depending on the values taken by $\mtt{i}$. 
Indeed, $\mtt{P_{exec}(i)}$ contains only the event related to the basic block identified by $\mtt{i}$.

\subsection{Modeling Attacker}
\label{sec:attacker-model}
We focus on test inversion and jump attacks.  
A test inversion attack is an attack where the result of a test is inverted. 
Whereas, a jump attack is an attack that directs the control flow of a program execution in a way that results in a path that does not exist inside the CFG of the program. 

Test inversion and jump attacks can be performed using physical means~\cite{pieee/Bar-ElCNTW06}, such as voltage and clock glitches, and electromagnetic and laser perturbations, to disturb program executions.
An attack can also result from transient errors or malicious software.  

We consider the multiple fault injection model for test inversion attacks, whereas we consider the single fault injection model for jump attacks. 
Indeed, the scenario where a jump from a basic block $\B$ into a basic block $\B'$ that is directly followed by a jump from $\B'$ into $\B$ may not be detected by our monitors. 
Note that the limitation of our monitors in detecting jump attacks in case of multiple fault injections is restricted to the case where the injections result in multiple jump attacks. 
Nevertheless, scenarios where there is only one jump attack and (possibly) other attacks, such as test inversion attack and event skip attack (provided that at most one occurrence of an event is skipped), can be detected by our monitors. 

Furthermore, we assume that the attacker can skip at most one of the two consecutive occurrences of an event. 
Otherwise, monitors may not receive all the necessary events to output correct verdicts. 

\paragraph*{Test inversion attack.}  
  Consider a basic block $\B$ that ends with a test $\mtt{A}$ = $\mtt{z:=(x \oprel y)}$, $\mtt{\mbox{\ttt{ifZ} } z \mbox{ \ttt{goto} } L}$. 
  There is a test inversion attack on $\mtt{A}$ when  $\mtt{B_T}$ is executed when $\mtt{(x \oprel y)=False}$, or when $\mtt{B_F}$ is executed when $\mtt{(x \oprel y)=True}$. 
  In practice, the result of $\mtt{A}$ can be inverted, for example, by:  
  \begin{itemize}
  \item 
  skipping the conditional branch instruction, so that $\mtt{B_T}$ is executed regardless whether $\mtt{(x \oprel y)}$ evaluates to $\mtt{True}$ or $\mtt{False}$),   
  \item 
  skipping the instruction that involves the logical expression provided that variable $\mtt{z}$ already holds the value that results in branch inversion, or 
  \item 
  flipping the value of $\mtt{z}$ after the logical expression being evaluated.  
  \end{itemize}

  \begin{definition}{\bf (Test Inversion Attack).}
  Let $\mtt{P}$ be a program, and let $\mtt{P_{exec}}$ = $\mtt{e_1, \ldots, e_n}$ be an execution of $\mtt{P}$.  
  We say that there is a test inversion attack on $\mtt{P_{exec}}$ if it violates $\mtt{R_1}$ or $\mtt{R_2}$ which are defined as follows, where $\mtt{i}$ identifies $\opreli$: 
  \begin{itemize}
  \item 
  $\mtt{R_1}$: for every $\mtt{j}$, if $\mtt{e_j} = \mtt{eT(i,x,y)}$   then $\mtt{(x \opreli y)= True}$. 
  \item
  $\mtt{R_2}$: for every $\mtt{j}$, if $\mtt{e_j} = \mtt{eF(i,x,y)}$ then $\mtt{(x \opreli y) = False}$.
  \end{itemize}  
  \end{definition}

\paragraph*{Jump attack.}   
  In our model, a jump attack interrupts an execution of a basic block, starts an execution of a basic block not at its first instruction, or results in an edge that does not exist in the CFG.    
  In practice, a jump attack can be performed, for example, by manipulating the target address of a branch or return. 
  Note that we do not consider intra-basic block jumps (which are equivalent to skipping one or more instruction inside the same basic block). 

  Let $\B$ be a basic block, and consider only events $\mtt{begin(i)}$ and $\mtt{end(i)}$.  
  Then, in the absence of jump attacks, an execution of $\B$ results in one of the following traces depending on whether none, one, or two events are skipped (assuming events duplication):   
  \begin{itemize}
    \item 
    $\mtt{tr_1=begin(i)}$.$\mtt{begin(i)}$.$\mtt{end(i)}$.$\mtt{end(i)}$, 
    \item 
    $\mtt{tr_2=begin(i)}$.$\mtt{end(i)}$.$\mtt{end(i)}$, 
    \item 
    $\mtt{tr_3=begin(i)}$.$\mtt{begin(i)}$.$\mtt{end(i)}$, or 
    \item 
    $\mtt{tr_4=begin(i)}$.$\mtt{end(i)}$. 
  \end{itemize}
  During a program execution, $\B$ may get executed more than once only if it is involved in a loop. In this case, between every two executions of $\B$ event $\mtt{reset(i)}$ should be emitted.

  \begin{definition}{\bf (Jump Attack).}\label{def:jump-attack}
  Let $\mtt{P}$ be a program, and let $\mtt{P_{exec}}$ be an execution of $\mtt{P}$. 
  Let $\mtt{P^J_{exec}(i)}$ = $\mtt{e_1, \ldots, e_n}$ denote the trace obtained from $\mtt{P_{exec}(i)}$ by filtering out all the events that are not in 
  $\Sigma_{\mtt{J}} =\{\mtt{begin(i)},\mtt{end(i)}, \mtt{reset(i)}\}$. 
  We say that there is a jump attack on $\mtt{P_{exec}}$ if there exists $\mtt{i}$ such that $\mtt{P^J_{exec}(i)}$ violates $\mtt{R_3}$, $\mtt{R_4}$ or $\mtt{R_5}$, which are defined as follows: 
  \begin{itemize}
  \item 
  $\mtt{R_3}$: for every $\mtt{j}$, if $\mtt{e_j} = \mtt{begin(i)}$ then 
    \begin{itemize}
      \item $\mtt{e_{j+1}} = \mtt{end(i)}$, if $\mtt{e_{j-1}} = \mtt{begin(i)}$. 
      \item $\mtt{e_{j+1}} = \mtt{end(i)}$, or $\mtt{e_{j+1}} = \mtt{begin(i)}$ and $\mtt{e_{j+2}} = \mtt{end(i)}$, if $\mtt{e_{j-1}} \neq \mtt{begin(i)}$. 
    \end{itemize}
  \item 
  $\mtt{R_4}$: for every $\mtt{j}$, if $\mtt{e_j} = \mtt{end(i)}$ then 
    \begin{itemize}
      \item $\mtt{e_{j-1}} = \mtt{begin(i)}$, if $\mtt{e_{j+1}} = \mtt{end(i)}$. 
      \item $\mtt{e_{j-1}} = \mtt{begin(i)}$, or $\mtt{e_{j-1}} = \mtt{end(i)}$ and $\mtt{e_{j-2}} = \mtt{begin(i)}$, if $\mtt{e_{j+1}} \neq \mtt{end(i)}$. 
    \end{itemize} 
  \item 
  $\mtt{R_5}$: there is no $\mtt{j}$ such that $\mtt{e_j} = \mtt{end(i)}$ and $\mtt{e_{j+1}} = \mtt{begin(i)}$.     
  \end{itemize}
  \end{definition}

  Definition~\ref{def:jump-attack} considers the jump attacks that result in executions that cannot be built by concatenating elements from \{$\mtt{tr_1}$, $\mtt{tr_2}$, $\mtt{tr_3}$, $\mtt{tr_4}$, $\mtt{reset(i)}$\}. 
  Namely, it considers the following attacks: 
  \begin{itemize}
  \item 
  Any attack that interrupts an execution of a basic block $\B$. This attack results in one or two consecutive occurrences of event $\mtt{begin(i)}$ that is not directly followed by event $\mtt{end(i)}$, which violates requirement $\mtt{R_3}$ of Definition~\ref{def:jump-attack}. 
  \item
  Any attack that starts the execution of a basic block $\B$ not from its beginning. This attack results in event $\mtt{end(i)}$ that is not directly preceded by event $\mtt{begin(i)}$, which violates requirement $\mtt{R_4}$ of Definition~\ref{def:jump-attack}. 
  \item
  Any attack that performs a backward jump from the end of a basic block $\B_2$ into the beginning of a basic block $\B_1$ (\ie the execution already went through $\B_1$ before reaching $\B_2$) such that there is no edge from $\B_2$ to $\B_1$ inside the related CFG.  
  This attack results in two executions of $\B_1$ that are not separated by, at least, an emission of event $\mtt{reset(i)}$. 
  Thus, it results in event $\mtt{end(i)}$ that is directly followed by event $\mtt{begin(i)}$, which violates requirement $\mtt{R_5}$ of Definition~\ref{def:jump-attack}. 
  Note that similar forward jumps are not considered by Definition~\ref{def:jump-attack} as they do not violate $\mtt{R_3}$, $\mtt{R_4}$ nor $\mtt{R_5}$. 
  \end{itemize}
  
  Finally, we note that a trace $\mtt{P^J_{exec}(i)}$ that starts with event $\mtt{reset(i)}$ or contains more than two consecutive occurrences of event $\mtt{reset(i)}$ does violate any of the requirements $\mtt{R_3}$, $\mtt{R_4}$ and $\mtt{R_5}$. 
  However, such a trace is produced only if there is a basic block $\B'$, with $\mtt{i' \neq i}$, that is executed not from its beginning. 
  Consequently, $\mtt{P^J_{exec}(i')}$ violates $\mtt{R_4}$ in this case, and thus the attack that can result in such situation is considered by Definition~\ref{def:jump-attack}. 

\begin{example}[Number of Events]
In order to check \ttt{verifyPIN} for test inversion attacks, it has to be instrumented to produce 8 events (4 $\mtt{bT(i,x,y)}$ events and 4 $\mtt{bF(i,x,y)}$ events) since \ttt{verifyPIN} contains two tests and every event has to be emitted twice. 
Whereas, to check \ttt{verifyPIN} for jump attacks, it has to be instrumented to produce 24 events (10 $\mtt{begin(i)}$ events and 14 $\mtt{end(i)}$ events) since \ttt{verifyPIN} has 5 basic blocks, contains two conditional branches, and every event has to be emitted twice.  
\end{example}

\section{Monitors}
\label{sec:monitors}
%
We propose monitors that check for the presence/absence of test inversion and jump attacks on a given execution.

We assume that each event is consecutively emitted twice.
Note that, in the absence of event skip attack, our monitors can still detect test inversion and jump attacks if each event is emitted only once\footnote{A smaller monitor can be used for jump attacks in this case, see Figure~\ref{fig:qea-jump}.}.
%
%
Note also that our monitors can be easily modified to report event skip attack by using a variable to count the number of received events or by tracing the visited states\footnote{An additional state has to be added to $\mtt{M_{TI}}$ in this case, see Figure~\ref{fig:qea-test-inversion}.}.
\subsection{A Monitor for Detecting Test Inversions}
%
Figure~\ref{fig:qea-test-inversion} depicts monitor $\mtt{M_{TI}}$, a QEA that detects test inversion attacks on a given execution. The alphabet of $\mtt{M_{TI}}$ is $\varSigma_{\mtt{TI}} =\{\mtt{eT(i,x,y)}$, $\mtt{eF(i,x,y)}\}$.
Monitor $\mtt{M_{TI}}$ has only one state, which is an accepting square state. It fails when event $\mtt{eT(i,x,y)}$ is emitted while $\mtt{(x \opreli y)= False}$ (\ie if the requirement $\mtt{R_1}$ is violated), or when event $\mtt{eF(i,x,y)}$ is emitted while $\mtt{(x \opreli y) = True}$ (\ie if the requirement $\mtt{R_2}$ is violated).
$\mtt{M_{TI}}$ accepts multiple occurrences of events $\mtt{eT(i,x,y)}$ and $\mtt{eF(i,x,y)}$ as long as the related guards hold.
Note that parameter $\mtt{i}$ is used to identify $\mtt{oprel^i}$ and it allows reporting the test that has been inverted in case of failure.

\begin{figure}[t]
  \centering
  \begin{tikzpicture}
    \tikzstyle{main}=[circle,thick,draw,fill=white,minimum size=6mm]
    \tikzstyle{failure}=[rectangle,thick,draw,fill=white,minimum size=6mm]
    \node[] (quan) at (-2.25,1.5)  {$\forall \;\mtt{i}$};
    \node[failure, initial, fill= gray] (1) at (0,0)  {1};
    \draw[->,>=latex',loop above] (1) to node {$\mtt{eT(i,x,y)}~\frac{\mathtt{(x \opreli y) == True}}{}$} (1);
    \draw[->, >=latex',loop right] (1) to node {$\mtt{eF(i,x,y)}~\frac{\mathtt{(x \opreli y) == False}}{}$}  (1);
    \draw[thick] ($ (current bounding box.south west) + (0,-0.25)$) rectangle ($(current
      bounding box.north east) + (0.25,0)$);
  \end{tikzpicture}
  \caption{$\mtt{M_{TI}}$, a QEA detecting test inversion attacks.} \label{fig:qea-test-inversion}
  \vspace{-1em}
\end{figure}
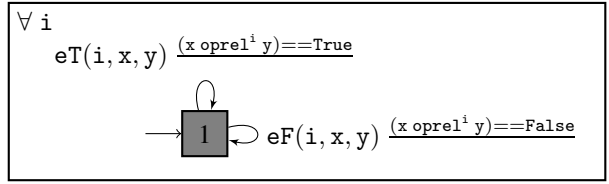

\begin{restatable}{proposition}{propositionTI}
  \label{prop:test_inversion}
  Let $\mtt{P}$ be a program and $\mtt{P_{exec}}$ an execution of $\mtt{P}$.
  Monitor $\mtt{M_{TI}}$ rejects $\mtt{P_{exec}}$ iff there is a test inversion attack on $\mtt{P_{exec}}$.
\end{restatable}


\begin{example}[Test Inversion Attack]
  \label{ex:ti-attack}
  An attacker can perform a test inversion attack on \ttt{verifyPIN} by skipping the second conditional branch: \ttt{ifZ \_t2 goto L2} (Line~10 of Listing~\ref{lst:verifyPIN-tac}).
  This attack allows the attacker to get authenticated with a wrong PIN.
  Assuming a wrong PIN, function \ttt{byteArrayCompare} returns a value \ttt{BOOL\_FALSE} (which is assigned to \ttt{\_t1}). Thus, value 0 is assigned to variable \ttt{\_t2} as a result of the logical instruction \ttt{\_t2:=(\_t1 == BOOL\_TRUE)}.
  At this point, if the conditional branch is skipped, the execution branches to $\B_3$ (the success branch) instead of $\B_4$ (the failure branch) which was supposed to be executed as the value of \ttt{\_t2} is 0.
  %
  %
  Provided that \ttt{verifyPIN} is instrumented as described in Section~\ref{sec:execution}, the faulted execution after filtering out any event that is not in $\Sigma_{\mtt{IT}}$ is as follows, where $\mtt{I_j}$ is the identifier of $\mtt{B_j}$, and initially we have \ttt{g\_ptc=3}:
  \begin{align*}
    \mtt{eT(I_1,3,0)}.\mtt{eT(I_1,3,0)} & .\mtt{eT(I_2,BOOL\_FALSE,BOOL\_TRUE)} \\
                                        & .\mtt{eT(I_2,BOOL\_FALSE,BOOL\_TRUE)}
  \end{align*}
  Note that $\mtt{eT(I_1,3,0)}$ is emitted at the beginning of $\B_2$, the success branch of the first test. It takes the arguments $\mtt{I_1}$, 3, and 0 since the corresponding test is inside $\B_1$, and the involved condition is $\mtt{(g\_ptc > 0)}$ where $\mtt{g\_ptc=3}$.
  On the other hand,  $\mtt{eT(I_2,BOOL\_FALSE},\mtt{BOOL\_TRUE})$ is emitted at the beginning of $\B_3$, the success branch of the second test. It takes the arguments $\mtt{I_2}$, \ttt{BOOL\_FALSE}, and \ttt{BOOL\_TRUE} since the corresponding test is inside $\B_2$, and the involved condition is $\mtt{(\_t1 == BOOL\_TRUE)}$ where \ttt{byteArrayCompare} returns \ttt{BOOL\_FALSE} into \ttt{\_t1} (as \ttt{g\_userPin} is a wrong PIN).

  The faulted execution is sliced by $\mtt{M_{IT}}$, based on the values that $\mtt{i}$ can take, into the following two slices:
  \begin{small}
    \begin{align*}
      \mtt{i}\mapsto \mtt{I_1} & : \mtt{eT(I_1,3,0)}.\mtt{eT(I_1,3,0)}                                       \\
      \mtt{i}\mapsto \mtt{I_2} & : \mtt{eT(I_2,BOOL\_FALSE,BOOL\_TRUE)}.\mtt{eT(I_2,BOOL\_FALSE,BOOL\_TRUE)}
    \end{align*}
  \end{small}
  Slice $\mtt{i}\mapsto \mtt{I_1}$ satisfies both requirements $\mtt{R_1}$ and $\mtt{R_2}$, and thus it is accepted by $\mtt{M_{TI}(I_1)}$.
  While, slice $\mtt{i}\mapsto \mtt{I_2}$ does not satisfy the requirement $\mtt{R_1}$ as event $\mtt{eT(I_2,BOOL\_FALSE,BOOL\_TRUE)}$ is emitted but $\mtt{(BOOL\_FALSE == BOOL\_TRUE)= False}$, and thus it is rejected by $\mtt{M_{IT}(I_2)}$.
  Indeed, the occurrence of $\mtt{eT(I_2,BOOL\_FALSE,BOOL\_TRUE)}$ leads into an implicit failure state since the related guard is not satisfied.
  Therefore, since slice $\mtt{i}\mapsto \mtt{I_2}$ is rejected by $\mtt{M_{IT}(I_2)}$, the whole faulted execution is rejected by $\mtt{M_{IT}}$.
\end{example}
%
\subsection{A Monitor for Detecting Jump Attacks}
%
Figure~\ref{fig:qea-jump} depicts monitor $\mtt{M_J}$, a QEA that detects jump attacks on a given program execution.
The alphabet of $\mtt{M_J(i)}$ is
$\varSigma_{\mtt{J}} =\{\mtt{begin(i)},\mtt{end(i)}, \mtt{reset(i)}\}$.
Monitor $\mtt{M_J}$ covers every basic block $\mtt{i}$ inside the CFG of the given program.
An instantiation of $\mtt{i}$ results in the EA $\mtt{M_J(i)}$.
Note that $\mtt{M_{J}}$ assumes a single fault injection model.

\begin{figure*}[h]
  \centering
  \begin{tikzpicture}
    \tikzstyle{main}=[circle,thick,draw,fill=white,minimum size=6mm]
    \tikzstyle{failure}=[rectangle,thick,draw,fill=white,minimum size=6mm]
    \node[] (quan) at (-1.5,1.25)  {$\forall \;\mtt{i}$};
    \node[failure, initial, fill= gray] (1) at (0,0)  {1};
    \node[failure] (2) at (3,0)  {2};
    \node[failure] (3) at (6,0)  {3};
    \node[failure, fill= gray] (4) at (9,0)  {4};
    \node[failure, fill= gray] (5) at (12,0) {5};
    %
    \draw[->,>=latex'] (1) to node [above] {$\mathtt{begin(i)}$} (2);
    \draw[->,>=latex'] (2) to node [above] {$\mathtt{begin(i)}$} (3);
    \draw[->,>=latex'] (3) to node [above] {$\mathtt{end(i)}$}  (4);
    \draw[->,>=latex'] (4) to node [above] {$\mathtt{end(i)}$} (5);
    \draw[->,>=latex'] (2) -- ++(0,1) |-  ++(6,0) node[above,xshift=-3cm] {$\mtt{end(i)}$}  -|  (4);
    \draw[->,>=latex',loop above] (1) to node {$\mtt{resetB(i)}$} (1);
    \draw[->,>=latex'] (4) -- ++(0,-0.75) |-  ++(-9,0) node[above,xshift=4.5cm] {$\mtt{reset(i)}$}  -|  (1.south);
    \draw[->,>=latex'] (5) -- ++(0,-1.25) |-  ++(-12.2,0) node[above,xshift=6cm] {$\mtt{reset(i)}$}  -|  ([xshift=-0.2cm]1.south);
    \draw[thick] ($ (current bounding box.south west) + (0,-0.25)$) rectangle ($(current
      bounding box.north east) + (0.25,0)$);
  \end{tikzpicture}
  \caption{$\mtt{M_{J}}$, a QEA detecting jump attacks.}
  \label{fig:qea-jump}
\end{figure*}
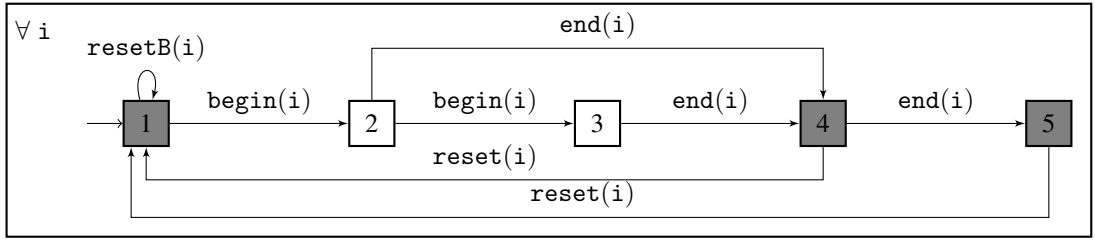

\begin{restatable}{proposition}{propositionJump}
  \label{prop:jump}
  Let $\mtt{P}$ be a program, and let $\mtt{P_{exec}}$ be an execution of $\mtt{P}$.
  Monitor $\mtt{M_J}$ rejects $\mtt{P_{exec}}$ iff there is a jump attack on $\mtt{P_{exec}}$. 
\end{restatable}

Indeed, $\mtt{M_{J}}$ cannot output a final verdict concerning $\mtt{R_1}$ until the end of the execution as event $\mtt{end(i)}$ may occur at any point in the future.
One way to explicitly catch the end of an execution is to include event \ttt{exit} in $\varSigma_{\mtt{J}}$, and add a transition from state (5), labeled by \ttt{exit}, to a new accepting square state, say state (6).
Then, an occurrence of event \ttt{exit} in state (3) means that event $ \mtt{end(i)}$ will definitely not occur, and thus leads to an implicit failure state.
A self-loop on state (4) labeled by \ttt{exit} is also required. 
Note that an occurrence of event \ttt{exit} at state (1) will also lead to failure.

\begin{example}[Jump Attack]
  \label{ex:j-attack}
  An attacker can perform a jump attack on \ttt{verifyPIN} by modifying the return address of the function \ttt{byteArrayCompare}, for example, into the address of the first instruction (\ttt{g\_ptc := 3}) of the basic block $\B_3$ (see Figure~\ref{fig:verifyPIN-cfg}).
  This attack interrupts the execution of $\B_2$, and allows the attacker to get authenticated with a wrong PIN as it skips the test that follows \ttt{byteArrayCompare}.
  Consequently, $\B_3$ (the success branch) will be executed regardless of the value returned by \ttt{byteArrayCompare}.
  Provided that \ttt{verifyPIN} is instrumented as described in Section~\ref{sec:execution}, the faulted execution after filtering out any event that is not in $\Sigma_{\mtt{J}}$ is as follows, where $\mtt{I_j}$ is the identifier of $\mtt{B_j}$:
  {\fontsize{8.7}{0} \selectfont
  \begin{align*}
     & \mtt{begin(I_1)}.\mtt{begin(I_1)}.\mtt{end(I_1)}.\mtt{end(I_1)}.\mtt{\bf begin(I_2)}.\mtt{\bf begin(I_2)}.\mtt{begin(I_3)} \\
     & .\mtt{begin(I_3)}.\mtt{end(I_3)}.\mtt{end(I_3)}.\mtt{begin(I_5)}.\mtt{begin(I_5)}.\mtt{end(I_5)}.\mtt{end(I_5)}
  \end{align*}
  }
  %

  The faulted execution is sliced by $\mtt{M_{J}}$, based on the values that $\mtt{i}$ can take, into the following four slices:
  \begin{align*}
    \mtt{i}\mapsto \mtt{I_1} & : \mtt{begin(I_1)}.\mtt{begin(I_1)}.\mtt{end(I_1)}.\mtt{end(I_1)},            \\
    \mtt{i}\mapsto \mtt{I_2} & : \mtt{begin(I_2)}.\mtt{begin(I_2)},                                          \\
    \mtt{i}\mapsto \mtt{I_3} & : \mtt{begin(I_3)}.\mtt{begin(I_3)}.\mtt{end(I_3)}.\mtt{end(I_3)}, \text{and} \\
    \mtt{i}\mapsto \mtt{I_5} & : \mtt{begin(I_5)}.\mtt{begin(I_5)}.\mtt{end(I_5)}.\mtt{end(I_5)}.
  \end{align*}
  Slices $\mtt{i}\mapsto \mtt{I_1}$, $\mtt{i}\mapsto \mtt{I_3}$, and $\mtt{i}\mapsto \mtt{I_5}$ satisfy  requirements $\mtt{R_3}$, $\mtt{R_4}$ and $\mtt{R_5}$.
  Thus, they are respectively accepted by $\mtt{M_{J}(I_1)}$, $\mtt{M_{J}(I_3)}$, and $\mtt{M_{J}(I_5)}$.
  However, slice $\mtt{i}\mapsto \mtt{I_2}$ does not satisfy the requirement $\mtt{R_3}$ as it contains two consecutive occurrences of event $\mtt{begin(I_2)}$ that are not followed by event $\mtt{end(I_2)}$. Thus it is rejected by $\mtt{M_{J}(I_2)}$.
  Indeed, the first occurrence of event $\mtt{begin(I_2)}$ fires the transition from state (1) into state (2), and the second occurrence of $\mtt{begin(I_2)}$ fires the transition from state (2) into state (3), see Figure~\ref{fig:qea-jump}.
  Thus, $\mtt{M_{J}(I_2)}$ ends in state (3), which is a failure state.
  Therefore, since slice $\mtt{i}\mapsto \mtt{I_2}$ is rejected by $\mtt{M_{J}(I_2)}$, the whole faulted execution is rejected by $\mtt{M_{J}}$.
\end{example}
Note, given a CFG $\mtt{(V,E)}$, monitor $\mtt{M_J}$ cannot detect an attack where a forward jump from the end of $\mtt{B\in V}$ into the beginning of $\mtt{B'\in V}$ with $\mtt{(B,B')\notin E}$ is executed.
Indeed, in order to detect such a jump, a global monitor with a structure similar to the CFG is needed where basic blocks are replaced by EAs that resemble $\mtt{M_{J}(i)}$ with the adjustment of reset transitions in accordance with the loops.

\section{Monitor Validation}
\label{sec:casestudy}
We validate our monitors by demonstrating their effectiveness in detecting simulated attacks against $\mtt{P}$~=~\ttt{verifyPIN}. 
Following the initial C implementation, we have implemented \ttt{verifyPIN} and the monitors using Java and AspectJ\footnote{\url{www.eclipse.org/aspectj/}}.
We have instrumented \ttt{verifyPIN} at the source code level. 
More precisely, for every required event we have defined an associated function which is called inside \ttt{verifyPIN} at the positions where the event has to be emitted as specified in Section~\ref{sec:execution}. 
The associated functions are used to define pointcuts in AspectJ. 
When a function is called, \ie a pointcut is triggered, the corresponding event is fed to the running monitor. The monitor then makes a transition based on its current state and the received event, and reports a verdict. 
The code segments executed by the monitor (called advices) are woven within the original source files to generate the final source code that is compiled into an executable. 

For example, Listing~\ref{lst:monitorTI} depicts the Java implementation of $\mtt{M_{TI}}$. 
Two states are defined: \ttt{Ok} (accepting state) and \ttt{Error}  (failure state).  
Function \ttt{updateState} (Lines~4-19) takes an event and then, after evaluating the condition $\mtt{(x \oprel y)}$, it updates the \ttt{currentState}. 
If the monitor is in state \ttt{Ok}, the state is updated into \ttt{Error} if the received event is \ttt{eT} and the condition evaluates to false, or the received event is \ttt{eF} and the condition evaluates to true. 
Once state \ttt{Error} is reached, the monitor cannot exit from it. 
Function \ttt{currentVerdict} (Lines~21-26) emits verdict \ttt{CURRENTLY\_TRUE} if the \ttt{currentState} is \ttt{Ok}, whereas it emits verdict \ttt{FALSE} if the \ttt{currentState} is \ttt{Error}. 

In what follows, we illustrate about how we carried out our  experiment and the obtained results. 
The experiment was conducted using Eclipse 4.11 and Java JDK 8u181 on a standard PC (Intel Core i7 2.2 GHz, 16 GB RAM).
%

\subsection{Normal Executions} 
Providing events duplication, running instrumented \ttt{verifyPIN} in the absence of an attacker results in one of the following 3 executions depending on the values of \ttt{g\_ptc} and \ttt{g\_userPin}: 
\begin{itemize}
  \item 
  $\mtt{P_{exec_1}}$ which contains 10 events: 4 events \ttt{begin}, 4 events \ttt{end} and 2 events \ttt{eF} that result from executing $\B_1$ and $\B_5$ (see Figure~\ref{fig:verifyPIN-cfg}).  
  This execution is performed when \ttt{g\_ptc} $\leq 0$.  
  \item 
  $\mtt{P_{exec_2}}$ which contains 20 events: 8 events \ttt{begin}, 8 events \ttt{end}, 2 events \ttt{eT} and 2 events \ttt{eF} that result from executing $\B_1$, $\B_2$, $\B_4$ and $\B_5$.   
  This execution is performed when \ttt{g\_ptc} $> 0$ and \ttt{g\_userPin} $\neq$ \ttt{g\_cardPin}. 
  \item 
  $\mtt{P_{exec_3}}$ which contains 20 events: 8 events \ttt{begin}, 8 events \ttt{end} and 4 events \ttt{eT} that result from executing $\B_1$, $\B_2$, $\B_3$ and $\B_5$.   
  This execution is performed when \ttt{g\_ptc} $> 0$ and \ttt{g\_userPin} = \ttt{g\_cardPin}. 
\end{itemize}

\begin{lstlisting}[%
  float=tp,%
  style=JStyle,%
  caption={Java implementation of $\mtt{M_{TI}}$.},%
  captionpos=t,%
  label=lst:monitorTI%
  ]
public class VerificationMonitorTI {
 private State currentState = State.Ok;

 public void updateState(Event e) {
  switch (this.currentState) {
   case Ok:
    int x = e.getX();
    int y = e.getY();
    String oprel = e.getOprel();
    boolean condition = evaluateCond(x,y,oprel);
    if ((e.getName().equals("eT") && !condition) || (e.getName().equals("eF") && condition))
     { this.currentState = State.Error; }
    break;
   case Error:
    // No need to execute any code.
    break;
  }
  System.out.println("moved to "+ this.currentState);
 }

 public Verdict currentVerdict () {
  switch(this.currentState) {
   case Ok: return Verdict.CURRENTLY_TRUE;
   case Error: return Verdict.FALSE;
   default: return Verdict.FALSE;
  } 
 } 
} 
\end{lstlisting} 

Table~\ref{tab:overhead} summarizes the cumulative execution time and the memory footprint of 100K runs of $\mtt{P_{exec_3}}$: (i) without instrumentation, (ii) with instrumentation, and (iii) with instrumentation and the monitors.  
Note that the memory consumption can be bounded as a monitor processes only one event at a time, keeps track only of the current state, and are parametrized by the current executing block.\footnote{We verified empirically that the memory consumption is insensitive to the number of events. However, we did not report the numbers because of lack of space.}  

The time overhead depends on the size of the application and the number of events. 
The size of \ttt{verifyPIN} is small; hence the measured overhead represents an extreme and unfavorable situation. 
A more representative measure of the overhead should be done on larger applications, especially because our approach aims at protecting the critical parts of an application instead of protecting every single instruction. 
Moreover, we note that our implementation is not optimized yet, and that using AspectJ for instrumentation is not the best choice performance-wise. Instrumentation causes most of the overhead (see Table~\ref{tab:overhead}). In the future, using tools such as ASM~\cite{Kuleshov07usingthe} to directly instrument the bytecode would result in a smaller overhead. 
\begin{table}[htbp]
  \caption{Cumulative execution time and memory footprint of 100K runs of $\mtt{P_{exec_3}}$.}
  \label{tab:overhead}
  \centering 
  \begin{tabular}{|c|c|c|}
    \hline  
    ~ &  \textbf{CPU Time (ms)} & \textbf{Memory (KB)} \\
    \hline \hline 
    \ttt{verifyPIN} & 164 & 284 \\ \hline 
    \rowcolor{gray!30}
    Inst. \ttt{verifyPIN} & 209 ($\times$ 1.27) & 796.6 ($\times$ 2.8) \\ \hline 
    Inst. \ttt{verifyPIN} \& Monitors & 228 ($\times$ 1.39) & 801 ($\times$ 2.82) \\  
  \hline 
\end{tabular}
\end{table}
%
\subsection{Test Inversion Attack} 
Function \ttt{verifyPIN} contains two tests.  
The first test: \ttt{\_t0 := (g\_ptc > 0)}, \ttt{ifZ \_t0 goto L1} is represented, in Java bytecode, using the instruction \ttt{ifle L1}. 
The instruction \ttt{ifle L1} compares \ttt{g\_ptc} and \ttt{0}, which are previously loaded into the stack, and performs a branch into \ttt{L1} if \ttt{g\_ptc} is less than or equal to 0.  
This test can be inverted by replacing \ttt{ifle L1} with \ttt{ifgt L1}, which performs a branch if \ttt{g\_ptc} is greater than 0.

Similarly, the second one: \ttt{\_t2 := (\_t1 == BOOL\_TRUE)}, \ttt{ifZ \_t2 goto L2} is represented using the instruction \ttt{if\_icmpne L2}, which compares \ttt{\_t1} and \ttt{BOOL\_TRUE}, and performs a branch into \ttt{L2} if they are not equal. 
This test can be inverted by replacing \ttt{if\_icmpne} with \ttt{if\_icmpeq}, which performs a branch if the operands are equal.

The binary opcodes of \ttt{ifle}, \ttt{ifgt}, \ttt{if\_icmpne} and \ttt{if\_icmpeq} are respectively ``1001 1110'', ``1001 1101'', ``1010 0000'' and ``1001 1111''.
Hence, replacing \ttt{ifle} with \ttt{ifgt} (\emph{resp.} \ttt{if\_icmpne} with \ttt{if\_icmpeq}) requires modifying 2 bits (\emph{resp.} 6 bits). 

Depending on the values of \ttt{g\_ptc} and \ttt{g\_userPin}, a test inversion attack can be used, e.g., to force authentication with a wrong PIN or to prevent the authentication with the correct PIN. 

As \ttt{verifyPIN} contains two tests, we consider the two following scenarios: 
\squishlist
  \item[\textbullet]
  In case of a wrong user PIN in the first trial (\ie \ttt{g\_ptc = 3}), inverting the second test results in a successful authentication. 
  This is the attack presented in Example~\ref{ex:ti-attack}. 
  %
  %
  Note that monitor $\mtt{M_{TI}}$ reports an attack after processing the first occurrence of event \ttt{eT} corresponding to the second test  as the related guard does not hold in this case. 
  That is after receiving 11 events: 4 events \ttt{begin}, 4 events \ttt{end} and 3 events \ttt{eT}. 
  Whereas, the full execution contains 20 events. 
  %
  \item[\textbullet]
  In case of a wrong user PIN in the fourth trial (\ie \ttt{g\_ptc = 0}), inverting both tests results in a successful authentication. 
  %
  %
  Again, $\mtt{M_{TI}}$ reports an attack after processing the first occurrence of event \ttt{eT}. 
  That is after receiving 5 events: 2 events \ttt{begin}, 2 events \ttt{end} and 1 events \ttt{eT}.
\squishend
Forcing the ``success branch'', in the scenarios above, can be also performed by replacing \ttt{ifle} and/or \ttt{if\_icmpne} with \ttt{nop}, which is equivalent to instruction skip, and thus results in the execution of the the ``success branch'' regardless of the operands' values. Replacing \ttt{ifle} (\emph{resp.} \ttt{if\_icmpne}) with \ttt{nop} (``0000 0000'', in binary) requires modifying 5 bits (\emph{resp.} 2 bits). 
Note that replacing \ttt{ifle} or \ttt{if\_icmpne} with \ttt{nop} only works with Java 6 or earlier\footnote
{
  Starting from Java 7, the typing system requires a stack map frame at the beginning of each basic block~\cite{JVM7}.
  Thus, a stack map frame is required by every branching instruction. 
  The stack map frame specifies the verification type of each operand stack entry and of each local variable. 
  Replacing \ttt{if\_icmpne} with \ttt{if\_icmpeq} does not result in a violation of the related stack map frame, however, replacing it with \ttt{nop} does. 
  Hence, in order to simulate the attack by replacing \ttt{if\_icmpne} with \ttt{nop}, the stack map frame also has to be modified. 
}. 

Our experiment showed that $\mtt{M_{TI}}$ can detect both attack  scenarios presented above.

\subsection{Jump Attack} 
A jump attack can be simulated, in Java bytecode, by replacing an instruction with \ttt{goto L} for a certain line number \ttt{L}.
However, this results in an inconsistent stackmap frame for Java 7 and latest versions. 
Nevertheless, it is possible to simulate the jump attack presented in Example~\ref{ex:j-attack} at the source code level
\footnote
{
  Indeed, it is not possible to simulate this attack by modifying the return address of \ttt{byteArrayCompare}.
  However, one can simulate the effect of \ttt{goto} using \ttt{break} and \ttt{continue} statements. 
}.
We note here that the main purpose is not performing the attack, but to validate that $\mtt{M_{J}}$ can detect jump attacks.
The latter is confirmed by our experiment. 
The faulted execution presented in Example~\ref{ex:j-attack} contains 16 events: 6 events \ttt{begin}, 6 events \ttt{end} and 2 events \ttt{eT}.
However, $\mtt{M_{J}}$ reports an attack after processing the first occurrence of event $\mtt{end}$ corresponding to $\B_3$.  
That is, after receiving 9 events: 4 events \ttt{begin}, 3 events \ttt{end} and 2 events \ttt{eT}. 
Note that the execution of $\B_2$ has been interrupted before, but $\mtt{M_{J}}$ cannot report an attack in this case until the end of the execution as event $\mtt{end}$ may appear at any time in the future.

\section{Related Work}
\label{sec:related-work}
This work introduces formal runtime verification monitors to detect fault attacks. 
Runtime verification/monitoring was successfully applied to several domains, {\eg} for monitoring financial transactions~\cite{ColomboP12}, monitoring IT logs~\cite{BasinCEHKM14}, monitoring electronic exams~\cite{fmsd/KassemFL17}, monitoring smart homes~\cite{El-HokayemF18a}.

In the following, we compare our work to the research endeavors that propose software-based protections against attacks. 
We distinguish between algorithm-level and instruction-level approaches.
%
\subsection{Algorithm-level Approaches}
At the algorithm level, there are approaches that use basic temporal redundancy~\cite{Robshaw97,eurocrypt/BonehDL97,ches/AumullerBFHS02,pieee/Bar-ElCNTW06, Ciet_practicalfault} such as computing a cryptographic operation twice then comparing the outputs. 
There are also approaches that use parity codes~\cite{ches/KarriKG03} and  digest values~\cite{fdtc/GenelleGP09}. 

Some other works make use of signature mechanisms and techniques to monitor executions, such as state automata and watchdog processor, in order to detect errors or protect the executions control flow. 
For instance, \cite{watchdog-85-267,DBLP:conf/itc/MahmoodME85,DBLP:journals/tc/SaxenaM90} use watchdogs for error detection at runtime. 
The underlying principle is to provide a watchdog processor~\cite{Lu1980WatchdogPA} with some information about the process (or processor) to be verified. Then, the watchdog concurrently collects the relevant information at runtime. 
An error is reported when the comparison test between the collected and provided information fails. 
The watchdog processor can also be used to detect control flow errors, as illustrated in~\cite{DBLP:conf/itc/MahmoodME85}, by verifying that the inserted assertions are executed in the order specified by the CFG. 
Ersoz~\etal~\cite{watchdog-85-267} propose the watchdog task, a software abstraction of the watchdog processor, to detect execution errors using some assertions.   
Saxena~\etal~\cite{DBLP:journals/tc/SaxenaM90} propose a control-flow checking method using checksums and watchdog.   
To ensure the control flow of a program, a signature is derived from the instructions based on checksum. 
However, these approaches focus on protecting solely basic blocks and the control flow, and do not protect branch instructions.  
Whereas our monitors detect attacks on branch instructions. 
Moreover, using a watchdog processor for monitoring involves larger communication overhead than using runtime verification. 

Nicolescu~\etal~\cite{DBLP:conf/dft/NicolescuSV03} propose a technique (SIED) to detect 
transient errors such as bit-flip.
This technique has been designed to be combined with an instruction duplication approach. 
It performs comparison checks and uses signatures to protect the intra-block and inter-block control flow, respectively. 
Relaying on comparison checks make it subject to fault injections that skip the check itself~\cite{iacr/BattistelloG15,fdtc/SchmidtH08}. 
Note that experiments have shown that SIED cannot detect all bit-flip faults. 
Later in~\cite{Nicolescu04}, the authors propose another error-detection mechanism that provides full coverage against single bit-flip faults.  
These works only consider single bit-flip as a fault model, whereas our definition of the attacks is independent of the technique used to perform the attack. 

Sere~\etal~\cite{Sr2011EvaluationOC} propose a set of countermeasures based on basic block signatures and security checks. 
The framework allows developers to detect mutant applications given a fault model, and thus developing secure applications. 
The countermeasures can be activated by the developer using an annotation
mechanism.
This approach requires some modifications in the Java Virtual Machine in order to perform the security checks. 
%
Bouffard~\etal~\cite{DBLP:conf/sscc/BouffardTL13} propose an automaton-based countermeasure against fault injection attacks. 
For a given program, every state of the corresponding automaton corresponds to a basic block of the CFG, and each transition corresponds to an edge, \ie allowed control flow. 
Thus the size of the resulting automaton is proportional to the size of the CFG of the program. Whereas, our monitors are lightweight and small in size. 

Fontaine~\etal~\cite{Fontaine05} propose a model to protect control flow integrity (CFI). 
The approach relies on instrumenting the LLVM IR code, and then using an external monitor (state automaton) which enforces CFI at runtime. 
Lalande~\etal~\cite{DBLP:conf/esorics/LalandeHB14} propose an approach to detect intra-procedural jump attacks at source code level, and to automatically inject some countermeasures. 
However, these approaches do not consider test inversion attacks.  

Algorithm-level approaches do not require changes to the instruction code.  
However, they are not effective against compile-time modifications.  
Moreover, it has been shown that they are not robust against multiple fault injections~\cite{ches/AumullerBFHS02,cosade/EndoHHTFA14} or skipping the critical parts of the code~\cite{iacr/BattistelloG15,fdtc/SchmidtH08}. 
Furthermore, most of the existing algorithm-level approaches are not generic.

Our approach is based on the formal model of QEAs, which is one of the most expressive and efficient form of runtime monitor~\cite{Bartocci2017}.
Our approach is also generic as it can be applied to any application and the monitors can be easily tweaked and used in combination with the monitors for the program requirements.
Moreover, monitors may run in a hardware-protected memory as they are lightweight and small in size. 
This provides more protection against synchronized multiple fault injections on both the monitored program and the monitor.   
Furthermore, runtime monitoring is modular and compatible with the existing approaches, for instance, monitor $\mtt{M_{TI}}$ can be used to detect fault attacks on the test that compares the outputs when an operation is computed twice.  

Note, to ensure the correct extraction of the necessary information (\ie events) from a running program, we use emission duplication.
This may require the duplication of every related instruction as duplication at the source code level may not be sufficient.
%
\subsection{Instruction-level Approaches}
At instruction-level, there are approaches that aim at providing fault-tolerance. These include (i) the approaches that apply the duplication or triplication of instructions~\cite{cases/BarenghiBKPR10,hipeac/BarryCR16} in order to provide 100\% protection against skip fault injections, and (ii) the approaches that rely on replacing every instruction with a sequence of functionally equivalent instructions such that skipping any of these instructions does not affect the outcome~\cite{jce/MoroHER14,jhss/Patranabis0M17}.  
Such approaches provide more guarantees than algorithm-level approaches.
Indeed, it is believed that they are robust against multiple fault injections under the assumption that skipping two consecutive instructions is too expensive and requires high synchronization capabilities~\cite{cases/BarenghiBKPR10,jce/MoroHER14}.  
However, they require dedicated compilers for code generation.
Moreover, approaches that apply the duplication or replacement of instructions are processor dependent as some instructions have to be replaced by a sequence of idempotent or functionally equivalent instructions. 
Furthermore, they introduce a large overhead in performance and footprint.
For instance, instruction duplication increases the overhead at least twice. 
Nevertheless, the overhead can be decreased by protecting only the critical parts of the code.  
In comparison, our monitors are easy to implement and deploy, introduce smaller overhead, and are independent of the processor and the complier.   
Note that, runtime monitoring does not provide 100\% protection against all fault attacks. 
Our monitors can detect test inversion and jump attacks.   
Providing more guarantees requires more monitors. 

There are also approaches that aim at ensuring CFI. 
Most existing CFI approaches follow the seminal work by Abadi \etal~\cite{tissec/AbadiBEL09}, which makes use of a special set of instructions in order to check the source and destination addresses involved in indirect jumps and indirect function calls. 
CFI approaches do not aim to provide a 100\% fault coverage.
Instead they aim at providing protections against jump-oriented attacks~\cite{cgo/ArthurMDA15,dimva/PayerBG15}, and return-oriented attacks~\cite{ccs/Shacham07}. 
A related technique called control flow locking (CFL) has been introduced by Bletsch \etal~\cite{acsac/BletschJF11} in order to provide protection against code-reuse attacks.
Instead of inserting checks at control flow transfers, CFL locks the corresponding memory address before a control flow transfer, and then unlocks it after the transfer is executed. 

Similar to other instruction-level approaches, these CFI approaches requires change to the instruction code, and usually introduce large overhead. 
Note that detecting jump attacks by our monitors is some sort of reporting CFI violations. 
Note also that CFI does not deal with test inversion attacks as taking any of the branches after a conditional branch does not violate CFI.

\section{Conclusions and Perspectives}
\label{sec:conclusion}
We formally define test inversion and jump attacks.
Then, we propose monitors expressed as Quantified Event Automata in order to detect these attacks. 
Our monitors are lightweight and small in size, and they support the duplication of events emission which provides protection against event skip attacks.   
Finally, we demonstrate the validity of our monitors using attack examples on \ttt{verifyPIN}.  

In the future, we will define more monitors to detect additional attacks following the principles exposed in this paper. For example, a monitor that can detect attacks on function calls. 
We also plan (i) to verify applications larger than \ttt{verifyPIN}, combined with detailed feasibility and performance analysis, (ii) to use a Java bytecode editing tool, such as ASM~\cite{Kuleshov07usingthe} or JNIF~\cite{Mastrangelo2014}, to simulate faults, and (iii) to deploy the monitors on hardware architectures such as smart cards, Raspberry Pi, and microcontrollers based on Arm Cortex-M processor. 
Furthermore, we consider building a tool for automatic generation of monitors from QEAs, and 
developing a runtime enforcement~\cite{FalconeMFR11,Falcone10,FalconeMRS18} framework where some corrective actions and countermeasures are automatically executed and taken respecticely once an attack is detected. 

\bibliographystyle{unsrt}  
\bibliography{biblio}

\appendix
\section{Proofs}
\label{sec:proofs}

 \propositionTI*

  \begin{proof}
  Assume that $\mtt{M_{TI}}$ rejects $\mtt{P_{exec}}$. 
  We have to show that there is a test inversion attack on $\mtt{P_{exec}}$, 
  \ie $\mtt{P_{exec}}$ violates $\mtt{R_1}$ or $\mtt{R_2}$.   
  As $\mtt{M_{TI}}$ rejects $\mtt{P_{exec}}$ then there exists $\mtt{i}$ such that $\mtt{M_{TI}(i)}$ fails (\ie ends in a failure state). 
  Thus, $\mtt{M_{TI}(i)}$ fires a transition into an implicit failure state since $\mtt{M_{TI}(i)}$ has only one (explicit) state which is an accepting state. 
  This means that $\mtt{P_{exec}}$ contains, for some $\mtt{x}$ and $\mtt{y}$, event $\mtt{eT(i,x,y)}$ such that $\mtt{(x \opreli y) = False}$ which violates $\mtt{R_1}$, or event $\mtt{eF(i,x,y)}$ such that $\mtt{(x \opreli y) = True}$ which violates $\mtt{R_2}$. 
  So, $\mtt{P_{exec}}$ violates $\mtt{R_1}$ or $\mtt{R_2}$, and thus there is a test inversion attack on $\mtt{P_{exec}}$. 
  Hence, we can conclude for the first direction. 

  To prove the second direction, we assume that there is a test inversion attack on $\mtt{P_{exec}}$, and we show that $\mtt{M_{TI}}$ rejects $\mtt{P_{exec}}$. 
  If there is a test inversion attack on $\mtt{P_{exec}}$, then $\mtt{P_{exec}}$ violates $\mtt{R_1}$ or $\mtt{R_2}$. 
  If $\mtt{P_{exec}}$ violates $\mtt{R_1}$ then it contains an event $\mtt{eT(i,x,y)}$ such that $\mtt{(x \opreli y) = False}$, which fires a transition into an implicit failure state as the guard related to event $\mtt{eT(i,x,y)}$ is not satisfied. Thus,  $\mtt{M_{TI}}$ fails and rejects $\mtt{P_{exec}}$. 
  If $\mtt{P_{exec}}$ violates $\mtt{R_2}$ then it contains an event $\mtt{eF(i,x,y)}$ such that $\mtt{(x \opreli y) = True}$, which fires a transition into an implicit failure state as the guard related to event $\mtt{eF(i,x,y)}$ is not satisfied. Thus,  $\mtt{M_{TI}}$ fails and rejects $\mtt{P_{exec}}$. 
  Hence, we conclude for the second direction and the proof is done. 
  \end{proof} 

  %

    \propositionJump*
    \begin{proof}
    Assume that $\mtt{M_{J}}$ rejects $\mtt{P_{exec}}$. 
    We have to show that there is a jump attack on $\mtt{P_{exec}}$, \ie there exists $\mtt{i}$ such that $\mtt{P^J_{exec}(i)}$ = $\mtt{e_1, \ldots, e_n}$ violates $\mtt{R_3}$, $\mtt{R_4}$ or $\mtt{R_5}$.  
    As $\mtt{M_{J}}$ rejects $\mtt{P_{exec}}$ then there exists $\mtt{i}$ such that EA $\mtt{M_{J}(i)}$ fails while consuming $\mtt{P^J_{exec}(i)}$. 
    We split according to the cases in which $\mtt{M_{J}(i)}$ fails.  
    \begin{itemize}
      \item 
      $\mtt{M_{J}(i)}$ has encountered event $\mtt{end(i)}$ in state (1). 
      This means that there exists $\mtt{j}$ such that $\mtt{e_j=end(i)}$, and that $\mtt{e_{j-1}}$ is either $\epsilon$ or $\mtt{reset(i)}$ since state (1) is the initial state and it can only be reached by $\mtt{reset(i)}$.  
      Thus, $\mtt{P^J_{exec}(i)}$ violates $\mtt{R_4}$.   
      \item 
      $\mtt{M_{J}(i)}$ ends in state (2). Then, in this case, there exists $\mtt{j}$ such that $\mtt{e_j=begin(i)}$ since state (2) can only be reached from state (1) through event $\mtt{begin(i)}$. 
      Moreover, we can deduce that $\mtt{e_{j+1}}$ is neither $\mtt{begin(i)}$ nor $\mtt{end(i)}$.   
      Thus, $\mtt{P^J_{exec}(i)}$ violates $\mtt{R_3}$. 
      \item 
      $\mtt{M_{J}(i)}$ has encountered event $\mtt{e_j=reset(i)}$ in state (2). State (2) can only be reached from state (1) through event $\mtt{begin(i)}$. 
      Then, $\mtt{e_{j-1}=begin(i)}$. 
      So, there exists $\mtt{e_{j-1}=begin(i)}$ such that $\mtt{e_j}$ is neither $\mtt{end(i)}$ nor $\mtt{begin(i)}$ which violates $\mtt{R_3}$.
      \item 
      $\mtt{M_{J}(i)}$ ends in state (3). Then there exists $\mtt{j}$ such that $\mtt{e_{j-1}=}$ $\mtt{e_j=begin(i)}$ because state (2) can only be reached from state (1) and state (3) can only be reached from state (2), both through event $\mtt{begin(i)}$. 
      Moreover, $\mtt{e_{j+1} \neq end(i)}$ as there is a transition from state (3) into state (4) labeled by $\mtt{end(i)}$, but $\mtt{M_{J}(i)}$ ends in state (3). 
      Thus, $\mtt{P^J_{exec}(i)}$ violates $\mtt{R_3}$. 
      \item 
      $\mtt{M_{J}(i)}$ has encountered event $\mtt{reset(i)}$ in state (3). This case is similar to the previous case and violates $\mtt{R_3}$. 
      \item 
      $\mtt{M_{J}(i)}$ has encountered  event $\mtt{begin(i)}$ in state (3). Then there exists $\mtt{j}$ such that $\mtt{e_{j-1}=e_j=e_{j+1}}$ = $\mtt{begin(i)}$, which violates $\mtt{R_3}$. 
      \item 
      $\mtt{M_{J}(i)}$ has encountered event $\mtt{begin(i)}$ in state (4).
      This means that there exists $\mtt{j}$ such that $\mtt{e_j=begin(i)}$, and $\mtt{e_{j-1}=end(i)}$ since state (4) can only be reached by $\mtt{end(i)}$ from state (2) or state (3), which can be reached only be reached by $\mtt{begin(i)}$.   
      Thus, $\mtt{P^J_{exec}(i)}$ violates $\mtt{R_5}$. 
      \item 
      $\mtt{M_{J}(i)}$ has encountered event $\mtt{end(i)}$ in state (5). 
      Then there exists $\mtt{j}$ such that $\mtt{e_{j-1}=e_j}$ = $\mtt{e_{j+1}=end(i)}$ because state (5) can only be reached from state (4) by $\mtt{end(i)}$, and state (4) can only be reached by $\mtt{end(i)}$. 
      Thus, we have that $\mtt{e_j=e_{j+1}=}$ $\mtt{end(i)}$ and $\mtt{e_{j-1} \neq begin(i)}$, which violates $\mtt{R_4}$.
      \item 
      $\mtt{M_{J}(i)}$ has encountered event $\mtt{begin(i)}$ in state (5). 
      This means that there exists $\mtt{j}$ such that $\mtt{e_{j}=begin(i)}$ and $\mtt{e_{j-1}=end(i)}$ since state (5) can only be reached by $\mtt{end(i)}$. Thus, $\mtt{P^J_{exec}(i)}$ violates $\mtt{R_5}$. 
    \end{itemize}
    Hence, there exists $\mtt{i}$ such that $\mtt{P^J_{exec}(i)}$ violates $\mtt{R_3}$, $\mtt{R_4}$ or $\mtt{R_5}$, and we can conclude about the first direction. 

    To prove the second direction, we assume that there exist $\mtt{i}$ such that $\mtt{P^J_{exec}(i)}$ violates $\mtt{R_3}$, $\mtt{R_4}$ or $\mtt{R_5}$ and we show that $\mtt{M_{J}}$ rejects $\mtt{P_{exec}}$. 
    %
    %
    We split cases according which requirement is violated. 
    \begin{itemize}
      \item If $\mtt{P^J_{exec}(i)}$ violates $\mtt{R_3}$, then it contains an event $\mtt{e_j=begin(i)}$ for some integer $\mtt{j}$ such that:  
      \begin{itemize}
        \item 
        $\mtt{e_{j-1}} = \mtt{begin(i)}$ and $\mtt{e_{j+1}} \neq \mtt{end(i)}$. 
          In this case if, before receiving $\mtt{e_{j-1}}$, $\mtt{M_{J}(i)}$ was:  
          \begin{itemize}
            \item 
            in state (1). 
            Then $\mtt{e_{j-1}} = \mtt{begin(i)}$ leads into state (2) and  $\mtt{e_{j}} = \mtt{begin(i)}$ leads into state (3). 
            As $\mtt{e_{j+1}} \neq \mtt{end(i)}$ and from state (3) there is only one explicit transition labeled by $\mtt{end(i)}$, then $\mtt{M_{J}(i)}$ ends in state (3) (\ie $\mtt{e_{j+1}}=\epsilon$) or a transition into an implicit failure state is fired.  
            \item 
            in state (2).  
            Then $\mtt{e_{j-1}} = \mtt{begin(i)}$ leads into state (3), and  $\mtt{e_{j}} = \mtt{begin(i)}$ leads into an implicit failure state. 
            \item 
            in state (3), (4) or (5). 
            Then $\mtt{M_{J}(i)}$ fails since none of the states (3), (4) and (5) has an explicit outgoing transition labeled by $\mtt{begin(i)}$. 
          \end{itemize}  
        \item 
        or $\mtt{e_{j-1}} \neq \mtt{begin(i)}$ and 
        $\mtt{e_{j+1}} \neq \mtt{end(i)}$, and ``$\mtt{e_{j+1}}\neq \mtt{begin(i)}$ or $\mtt{e_{j+2}} \neq\mtt{end(i)}$''.  
        This is equivalent to $\mtt{e_{j-1}} \neq \mtt{begin(i)}$ and,  
        $\mtt{e_{j+1}} = \epsilon$ (\ie $\mtt{j=n}$) or $\mtt{e_{j+1}} = \mtt{reset(i)}$ or  ``$\mtt{e_{j+1}} = \mtt{begin(i)}$ and $\mtt{e_{j+2}} \neq\mtt{end(i)}$'' since $\Sigma_{\mtt{J}} =\{\mtt{begin(i)},\mtt{end(i)}, \mtt{reset(i)}\}$. 
        In this case if, before receiving $\mtt{e_{j-1}}$, $\mtt{M_{J}(i)}$ was: 
        \begin{itemize}
          \item 
          in state (1). If $\mtt{e_{j-1}}$ is $\mtt{end(i)}$, a transition into an implicit failure state is fired. 
          Otherwise, we have that $\mtt{e_{j-1}}$ is $\mtt{reset(i)}$ or $\epsilon$ (\ie $\mtt{j=1}$), and thus $\mtt{M_{J}(i)}$ stays in state (1).  
          Then, $\mtt{e_j=begin(i)}$ leads into state (2). %
          Then in state (2), if $\mtt{e_{j+1}}$ is $\mtt{reset(i)}$ a transition into an implicit failure state is fired; 
          if $\mtt{e_{j+1}}$ is $\epsilon$ then $\mtt{M_{J}(i)}$ ends in failure state (2);  
          if $\mtt{e_{j+1}}$ is $\mtt{begin(i)}$ the transition from state (2) into state (3) is fired and, as $\mtt{e_{j+2}} \neq\mtt{end(i)}$ in this case, then a transition into an implicit failure state is fired or $\mtt{M_{J}(i)}$ ends in failure state (3).  
          \item 
          in state (2). If $\mtt{e_{j-1}}$ is $\mtt{reset(i)}$, a transition into an implicit failure state is fired. 
          Otherwise, we have that $\mtt{e_{j-1}} = \mtt{end(i)}$ which leads from state (2) into state (4). Then in state (4), as $\mtt{e_{j}} = \mtt{begin(i)}$, a transition into an implicit failure state is fired. 
          %
           \item 
          in state (3). Similar to the case of state (2).
          %
          %
          \item 
          in state (4). If $\mtt{e_{j-1}}$ is $\mtt{end(i)}$, the transition into state (5) is fired. Then in state (5), as $\mtt{e_{j}} = \mtt{begin(i)}$, a transition into an implicit failure state is fired.
          %
          %
          Otherwise, we have that $\mtt{e_{j-1}} = \mtt{reset(i)}$, and thus the transition from state (4) into state (1) is fired. 
          Then in state (1), as $\mtt{e_{j}} = \mtt{begin(i)}$, the transition into state (2) is fired. 
          Then in state (2), if $\mtt{e_{j+1}}$ is $\mtt{reset(i)}$ a transition into an implicit failure state is fired; 
          if $\mtt{e_{j+1}}$ is $\epsilon$ then $\mtt{M_{J}(i)}$ ends in state (2) which is a failure state; 
          if $\mtt{e_{j+1}}$ is $\mtt{begin(i)}$ the transition from state (2) into state (3) is fired and, as $\mtt{e_{j+2}} \neq\mtt{end(i)}$ in this case, then a transition into an implicit failure state is fired or $\mtt{M_{J}(i)}$ ends in state (3), which is a failure state. 
          \item 
          in state (5). If $\mtt{e_{j-1}}$ is $\mtt{reset(i)}$ the reasoning is similar to the case of state (4) when $\mtt{e_{j-1}} = \mtt{reset(i)}$.  
          Otherwise, we have that $\mtt{e_{j-1}} = \mtt{end(i)}$ which leads into an implicit failure state.  
          %
        \end{itemize}
      \end{itemize}
      Hence, If $\mtt{P^J_{exec}(i)}$ violates $\mtt{R_3}$ then $\mtt{M_{J}(i)}$ fails, and thus $\mtt{M_{J}}$ rejects $\mtt{P_{exec}}$. 
      %
      \item If $\mtt{P^J_{exec}(i)}$ violates $\mtt{R_4}$, then it contains an event $\mtt{e_j=end(i)}$ for some integer $\mtt{j}$ such that 
      \begin{itemize}
        \item 
        $\mtt{e_{j-1}} \neq \mtt{begin(i)}$ and $\mtt{e_{j+1}} = \mtt{end(i)}$. 
        In this case if, before receiving $\mtt{e_{j-1}}$, $\mtt{M_{J}(i)}$ was:  
        \begin{itemize}
          \item 
          in state (1). If $\mtt{e_{j-1}}$ is $\mtt{reset(i)}$ the self-loop transition over state (1) is fired. Then, as $\mtt{e_j=end(i)}$, a transition into an implicit failure state is fired.  
          Otherwise, if $\mtt{e_{j-1}}$ is $\mtt{end(i)}$, then a transition into an implicit failure state is fired. 
          Otherwise, we have that $\mtt{e_{j-1}} = \epsilon$ (\ie $\mtt{j=1}$), and thus $\mtt{e_j=end(i)}$ leads from state (1) into an implicit failure state.   
          \item 
          in state (2). If $\mtt{e_{j-1}}$ is $\mtt{reset(i)}$ a transition into an implicit failure state is fired.  
          Otherwise, we have that $\mtt{e_{j-1}}=\mtt{end(i)}$, and thus  the transition from state (2) into state (4) is fired. 
          Then in state (4), event $\mtt{e_{j}} = \mtt{end(i)}$ leads into state (5). 
          Then in state (5), event $\mtt{e_{j+1}} = \mtt{end(i)}$ leads into an implicit failure state. 
          \item 
          in state (3). Similar to the case of state (2).
          \item 
          in state (4). If $\mtt{e_{j-1}} = \mtt{reset(i)}$ the transition to state (1) is fired. Then in state (1), $\mtt{e_j=end(i)}$leads into an implicit failure state.  
          Otherwise, we have that $\mtt{e_{j-1}} = \mtt{end(i)}$ which fires  the transition from state (4) into state (5). 
          Then in state (5), $\mtt{e_{j}} = \mtt{end(i)}$ leads into an implicit failure state. 
          \item 
          in state (5). If $\mtt{e_{j-1}} = \mtt{reset(i)}$ the transition to  state (1) is fired. Then $\mtt{e_j=end(i)}$ leads into an implicit failure state. 
          Otherwise, we have $\mtt{e_{j-1}} = \mtt{end(i)}$ which fires a transition into an implicit failure state.  
        \end{itemize} 
        \item 
        or $\mtt{e_{j+1}} \neq \mtt{end(i)}$ and 
        $\mtt{e_{j-1}} \neq \mtt{begin(i)}$, and ``$\mtt{e_{j-1}}\neq \mtt{end(i)}$ or $\mtt{e_{j-2}} \neq\mtt{begin(i)}$''.  
        This is equivalent to $\mtt{e_{j+1}} \neq \mtt{end(i)}$, and 
        $\mtt{e_{j-1}} = \mtt{reset(i)}$ or  ``$\mtt{e_{j-1}} = \mtt{end(i)}$ and $\mtt{e_{j-2}} \neq\mtt{begin(i)}$'' since $\Sigma_{\mtt{J}} =\{\mtt{begin(i)},\mtt{end(i)}, \mtt{reset(i)}\}$. 
        In this case if, before receiving $\mtt{e_{j-1}}$, $\mtt{M_{J}(i)}$ was: 
        \begin{itemize}
          \item 
          in state (1). If $\mtt{e_{j-1}}$ is $\mtt{end(i)}$, a transition into an implicit failure state is fired.
          Otherwise, we have that $\mtt{e_{j-1}}$ is $\mtt{reset(i)}$ or 
          $\epsilon$ (\ie $\mtt{j=1}$), and thus $\mtt{M_{J}(i)}$ stays in state (1). Then, as $\mtt{e_j=end(i)}$, a transition into an implicit failure state is fired.  
          \item 
          in state (2). In this case $\mtt{e_{j-1}}$ must be equal to $\mtt{reset(i)}$ since if $\mtt{e_{j-1}} = \mtt{end(i)}$, then 
          $\mtt{e_{j-2}} \neq \mtt{begin(i)}$ and thus state (2) cannot be reached. 
          Indeed, state (2) can only be reached from state (1) by $\mtt{begin(i)}$. 
          In state (2), $\mtt{e_{j-1}} = \mtt{reset(i)}$ leads into an implicit failure state. 
          \item 
          in state (3). Similar to the case of state (2).
          \item 
          in state (4). If $\mtt{e_{j-1}}$ is $\mtt{reset(i)}$ the transition into state (1) is fired. 
          Then in state (1), as $\mtt{e_{j}} = \mtt{end(i)}$, a transition into an implicit failure state is fired. 
          Otherwise, we have that $\mtt{e_{j-1}} = \mtt{end(i)}$ which leads into state (5).  
          Then in state (5), $\mtt{e_{j}} = \mtt{end(i)}$ leads into an implicit failure state. 
          \item 
          in state (5). If $\mtt{e_{j-1}} = \mtt{reset(i)}$ the transition into state (1) is fired. 
          Then, as $\mtt{e_{j}} = \mtt{end(i)}$, a transition into an implicit failure state is fired. 
          Otherwise, we have that $\mtt{e_{j-1}} = \mtt{end(i)}$ which leads into an implicit failure state.  
        \end{itemize}
      \end{itemize}
      Hence, If $\mtt{P^J_{exec}(i)}$ violates $\mtt{R_4}$ then $\mtt{M_{J}(i)}$ fails, and thus $\mtt{M_{J}}$ rejects $\mtt{P_{exec}}$. 
      %
      \item If $\mtt{P^J_{exec}(i)}$ violates  $\mtt{R_5}$, then it contains an event $\mtt{e_j=end(i)}$ for some integer $\mtt{j}$ such that $\mtt{e_{j+1}=begin(i)}$. In this case if, before receiving $\mtt{e_j}$, $\mtt{M_{J}(i)}$ was: 
      \begin{itemize}
        \item 
        in state (1). Then event $\mtt{e_j=end(i)}$ fires a transition into an implicit failure state.  
        \item 
        in state (2). Then event $\mtt{e_j=end(i)}$ fires the transition into state (4).   
        In state (4), $\mtt{e_{j+1}} = \mtt{begin(i)}$ fires a transition into an implicit failure state.  
        \item 
        in state (3). Similar to the case of state (2). 
        \item 
        in state (4). Then $\mtt{e_j=end(i)}$ fires the transition into state (5).   
        In state (5), $\mtt{e_{j+1}} = \mtt{begin(i)}$ fires a transition into an implicit failure state. 
        \item 
        in state (5). Then $\mtt{e_j=end(i)}$ fires a transition into an implicit failure state.  
      \end{itemize}
      Hence, if $\mtt{P^J_{exec}(i)}$ violates $\mtt{R_5}$ then $\mtt{M_{J}(i)}$ fails, and thus $\mtt{M_{J}}$ rejects $\mtt{P_{exec}}$. 
    \end{itemize}
    Therefore, if there exists $\mtt{i}$ such that $\mtt{P^J_{exec}(i)}$ violates $\mtt{R_3}$, $\mtt{R_4}$ or $\mtt{R_5}$ then $\mtt{M_{J}}$ rejects $\mtt{P_{exec}}$. Thus, we can conclude for the second direction and the proof is done. 
  \end{proof} 

\end{document}